\documentclass[reqno,11pt]{amsart}

\usepackage{amsmath,amsthm,amssymb,amsfonts,epsfig,color,graphicx,enumerate,accents,subfigure,mathtools}
\usepackage{hyperref}
\usepackage{IEEEtrantools}
\usepackage[a4paper,margin=3truecm]{geometry}
\usepackage[utf8]{inputenc}

\theoremstyle{plain}
\newtheorem{thm}{Theorem}[section]
\newtheorem{lemma}[thm]{Lemma}
\newtheorem{coro}[thm]{Corollary}
\newtheorem{prop}[thm]{Proposition}
\newtheorem{defi}[thm]{Definition}
\theoremstyle{remark}
\newtheorem{rema}[thm]{\bf Remark}

\numberwithin{equation}{section}

\newenvironment{ieee*}[1]{\begin{IEEEeqnarray*}{#1}}{\end{IEEEeqnarray*}\ignorespacesafterend}

\newcommand{\ie}{\emph{i.e.}}
\newcommand{\eqdef}{\vcentcolon=}

\newcommand{\R}{\mathbb{R}}
\newcommand{\Z}{\mathbb{Z}}
\newcommand{\N}{\mathbb{N}}
\newcommand{\C}{\mathbb{C}}

\newcommand{\eps}{\varepsilon}
\newcommand{\Prob}{\mathcal{P}}
\newcommand{\rnd}{(\R^d)^N}

\newcommand{\gra}[1]{\left \{ #1 \right \}}
\newcommand{\st}{\left | \right.}
\newcommand{\abs}[1]{\left| #1 \right|}
\newcommand{\norm}[1]{\left \lVert {#1} \right \rVert}

\newcommand{\onab}{\nabla_\lambda}
\newcommand{\Con}{\mathcal{C}}

\DeclareMathOperator{\Lip}{Lip}

\title{From wave-functions to single electron densities}

\author{U. Bindini}
\address{Scuola Normale Superiore \\ Piazza dei Cavalieri 7, 56126 Pisa \\ Italy}
\email{ugo.bindini@sns.it}

\author{L. De Pascale}
\address{Dipartimento di Matematica e Informatica, Università di Firenze \\ Viale Morgagni 67/a, 50134 Firenze \\ Italy}
\email{luigi.depascale@unifi.it}
\urladdr{http://web.math.unifi.it/users/depascal/}

\date{\today}

\begin{document}

\begin{abstract}
We investigate some of the properties of the mapping from wave-functions to single particle densities, partially answering an open question posed by E. H. Lieb in 1983.
\end{abstract}

\maketitle




\section{Introduction}

Consider a system of $N$ charged particles, interacting with each other due to the Coulomb force. The coordinates to describe the system will be $N$ space-spin variables $z_1,\dotsc, z_N$, where $z_j = (x_j,s_j)$, with $x_j \in \R^d$, $s_j \in \gra{0,\dotsc,q}$. Here $x_j$ represents the position of the $j$-th particle, $s_j$ its spin. Most of the time one can consider $d=3$, the physical case, but we do not want to limit ourselves to this specific value, so we keep a generic dimension of the space. We will often group the space variables as $X = (x_1,\dotsc,x_N) \in (\R^d)^N$.

The description of the system will be given by a complex-valued wave-function $\psi(z_1,\dotsc,z_N)$, belonging to $H^1((\R^d \times \Z_q)^N;\C)$ and always normalized with
\[
\sum_{s_1,\dotsc,s_N} \int_{(\R^d)^N} \abs{\psi(x_1,s_1,\dotsc,x_N,s_N)}^2 dX = 1,
\]
which may be viewed in the following way (Born interpretation): $\abs{\psi(x_1,s_1,\dotsc,x_N,s_N)}^2$ is the probability density that the particles are in the positions $x_j$ with spin $s_j$.

If the system is composed of bosons, Bose-Einstein statistics apply, and hence the wave-function $\psi$ must be symmetric, in the sense that
\[
\psi(z_1,\dotsc,z_N) = \psi(z_{\sigma(1)},\dotsc,z_{\sigma(N)})
\]
for every permutation $\sigma \in \mathfrak{S}_N$. Thus we have a class of \textit{bosonic} wave-functions given by
\[
\mathcal{S} = \gra{ \psi \st \psi \in H^1((\R^d\times \Z_q)^N), \psi \text{ is symmetric}}.
\]

On the other hand, one can consider a system of fermions, obeying Fermi-Dirac statistics, \ie,
\[
\psi(z_1,\dotsc,z_N) = \mathrm{sign} (\sigma) \psi(z_{\sigma(1)},\dotsc,z_{\sigma(N)})
\]
for every permutation $\sigma \in \mathfrak{S}_N$. Thus we have the class of \textit{fermionic} wave-functions given by
\[
\mathcal{A} = \gra{ \psi \st \psi \in H^1((\R^d\times \Z_q)^N), \psi \text{ is antisymmetric}}.
\]

If $\psi$ is a wave-function, we introduce the single particle density
\[
\rho[\psi](x) = \sum_{s_1, \dotsc, s_N} \int_{(\R^d)^{N-1}} \abs{\psi(x,s_1,x_2,s_2,\dotsc,x_N,s_N)}^2 dx_2 \dotsm dx_N,
\]
where, by symmetry, the integral might be done with respect to any choice of $N-1$ space variables. To consider the single particle density, rather than the whole wave-function, is quite natural from the physical point of view, since charge density is a fundamental quantum-mechanical observable, directly obtainable from experiment. 

We can then consider the maps
\[
\begin{array}{rccc}
\Phi^{\mathcal{S}} \colon & \mathcal{S} & \to & \mathcal{P}(\R^d) \\
& \psi & \mapsto & \rho[\psi]
\end{array}
\quad \text{and} \quad
\begin{array}{rccc}
\Phi^{\mathcal{A}} \colon & \mathcal{A} & \to & \mathcal{P}(\R^d) \\
& \psi & \mapsto & \rho[\psi],
\end{array}
\]
where $\Prob(\R^d)$ denotes the space of probability measures over $\R^d$. The properties of these maps have been studied \textit{in primis} by Lieb \cite{lieb2002density}. In particular we have
\begin{itemize}
  \item the range of $\Phi^{\mathcal{S}}$ and $\Phi^{\mathcal{A}}$ is exactly the set
  \[
  \mathcal{R} = \gra{\rho \in L^1(\R^d) \st \rho \geq 0, \sqrt{\rho} \in H^1(\R^d), \int_{\R^d} \rho(x) dx = 1};
  \]
  
  \item $\Phi^{\mathcal{S}}$ and $\Phi^{\mathcal{A}}$ are continuous with respect to the $H^1$ norms: if $\psi_k \to \psi$ in $H^1$, then $\sqrt{\rho[\psi_k]} \to \sqrt{\rho[\psi]}$ in $H^1$.
\end{itemize}

It is quite clear that the maps $\Phi^{\mathcal{S}}$ and $\Phi^{\mathcal{A}}$ are not invertible, even restricting the codomain to $\mathcal{R}$ --- in fact, different $\psi$'s may have the same single particle density. However, suppose that $(\sqrt{\rho_k})_{k \geq 1}$ converge to $\sqrt{\rho}$ in $H^1$, and take $\psi$ such that $\rho = \rho[\psi]$. Can we find $(\psi_k)_{k \geq 1}$ such that $\rho_k = \rho[\psi_k]$ and $\psi_k \to \psi$ in $H^1$? In other words, are $\Phi^{\mathcal{A}}$ and $\Phi^{\mathcal{S}}$ open? This problem, to our knowledge first stated in \cite[Question 2]{lieb2002density}, is still open.

In this paper we will give a partial positive answer, proving the following results when $q=1$ (\ie, without taking into account the spin variables).

\begin{thm} \label{main-thm-1}
  Let $\psi \in \mathcal{S}$ non-negative. Given $(\rho_n)_{n \geq 1}$ such that $\sqrt{\rho_n} \to \sqrt{\rho[\psi]}$ in $H^1(\R^d)$, there exist $(\psi_n)_{n \geq 1}$ symmetric and non-negative such that $\rho_n = \rho[\psi_n]$ and $\psi_n \to \psi$ in $H^1((\R^d)^N)$.
\end{thm}

\begin{thm} \label{main-thm-2}
  Let $\psi \in \mathcal{S}$ real-valued. Given $(\rho_n)_{n \geq 1}$ such that $\sqrt{\rho_n} \to \sqrt{\rho[\psi]}$ in $H^1(\R^d)$, there exist $(\psi_n)_{n \geq 1}$ symmetric and complex-valued such that $\rho_n = \rho[\psi_n]$ and $\psi_n \to \psi$ in $H^1((\R^d)^N)$.
\end{thm}

Notice that the first result is already of physical interest, since in many cases the ground state of a system of $N$-particles is non-negative.

\subsection*{Outline of the paper}

The main tools will be the smoothing of transport plans as introduced and studied in \cite{bindini2019smoothing, bindini2017optimal} and an application of the weighted Sobolev spaces. In Section \ref{L2-sequence} we will start by constructing explicitly an $L^2$ approximation of $\abs{\psi}$ which respects the marginal constraint. Then in Section \ref{H1-regularity} we regularize it in order to obtain a Sobolev regular sequence which converges in $H^1((\R^d)^N)$ to $\abs{\psi}$ and still maintains the marginal constraint. This will complete the proof of Theorem \ref{main-thm-1}. In the final Section \ref{sign-section}, making use of a suitable weighted Sobolev space, we show how to deal with the sign of the wave-function, finally proving the main result in its completeness.

\subsection*{Acknoledgements}
The authors wish to thank Simone Di Marino, Gero Friesecke, Paola Gori-Giorgi and Mathieu Lewin for the useful conversations on the topic of this paper. Part of the work and a preliminary presentation of it was done at the Banff International Research Station. The first author is grateful to the Scuola Normale Superiore (Pisa), and the second author is grateful to the Università di Firenze for the financial support.


\section{Construction of $L^2$ wavefunctions} \label{L2-sequence}

In this section we start the construction by proving the following

\begin{thm} \label{main-thm}
 Let $\rho_n, \rho \in L^1(\R^d)$ such that $\sqrt{\rho_n} \to \sqrt{\rho}$ in $L^2(\R^d)$, $\int \rho = \int \rho_n = 1$ and let $\varphi \in L^2(\rnd)$ symmetric, $\varphi \geq 0$, such that $\rho = \rho[\varphi]$. Then there exists a sequence $(\varphi_n) \subseteq L^2(\rnd)$ such that $\varphi_n$ is symmetric, $\rho_n = \rho[\varphi_n]$ and $\varphi_n \to \varphi$ in $L^2(\rnd)$.
\end{thm}

For fixed $n \in \N$, let $\sigma^0_n = \rho$ and $\varphi^0_n = \varphi$ and define inductively for $k \geq 0$
\begin{align*}
 E_n^k &= \gra{x \in \R^d \st \sigma_n^k(x) > \rho_n(x)}, \\
 S_n^k(X) &= \frac{1}{N} \sum_{j = 1}^N \frac{\sigma_n^k(x_j) - \rho_n(x_j)}{\sigma_n^k(x_j)} \chi_{E_n^k}(x_j), \\
 \varphi_n^{k+1}(X) &= \varphi_n^k(X) \sqrt{1 - S_n^k(X)}, \\
 \sigma_n^{k+1}(x) &= \int \varphi_n^{k+1}(x,x_2,\dotsc,x_N)^2 dx_2 \dotsm dx_N.
\end{align*}

Notice that, for every $k,n$, the function $\varphi^k_n$ is symmetric. The sequence $(\varphi^k_n)_{k \geq 0}$ is monotone decreasing, as proved in the following

\begin{lemma} \label{monotonicity-lemma}
 \begin{enumerate}[(i)]
  \item $0 \leq \varphi_n^{k+1} \leq \varphi_n^{k} \leq \varphi$;
  \item $0 \leq \sigma_n^{k+1} \leq \sigma_n^k \leq \rho$;
  \item $E_n^{k+1} \subseteq E_n^{k} \subseteq E_n^0$.
 \end{enumerate}
\end{lemma}

\begin{proof} Since $0 \leq S_n^k(X) \leq 1$, the factor $\sqrt{1 - S_n^k(X)}$ is less or equal than 1, and the inequalities in (i) are obvious; (ii) and (iii) follow.
\end{proof}

In order to estimate some $L^2$ norms which will appear later, the following lemma will also prove useful.

\begin{lemma} \label{int-E-lemma} If $k \geq 0$ and $E \subseteq E_n^k$, then
\[
 \int_{E} (\sigma_n^k(x) - \rho_n(x)) dx \leq \left( \frac{N-1}{N} \right)^k \int_{E} (\rho(x) - \rho_n(x)) dx.
\] 
\end{lemma}

\begin{proof} By induction on $k$. For $k = 0$ the inequality is in fact an equality.

Suppose now the thesis is true for $k$, and fix $E \subseteq E_n^{k+1}$. Using the fact that $E_n^{k+1} \subseteq E_n^k$ one has

\begin{align*}
 \int_{E} (\sigma_n^{k+1}(x) - \rho_n(x)) dx &= \int_{E \times (\R^d)^{N-1}} \varphi_n^{k}(X)^2 (1 - S_n^{k}(X)) dX - \int_E \rho_n(x) dx \\
  &= \int_{E} (\sigma_n^{k}(x) - \rho_n(x)) dx - \int_{E \times (\R^d)^{N-1}} \varphi_n^k(X)^2 S_n^{k}(X) dX.
\end{align*}

Notice that
\begin{align*}
 \int_{E \times (\R^d)^{N-1}} \varphi_n^k(X)^2 S_n^{k}(X) dX &= \frac{1}{N} \sum_{j = 1}^N \int_{E \times (\R^d)^{N-1}} \varphi_n^k(X)^2 \frac{\sigma_n^k(x_j) - \rho_n(x_j)}{\sigma_n^k(x_j)} \chi_{E_n^k}(x_j) dX \\
  & \geq \frac{1}{N} \int_{E \times (\R^d)^{N-1}} \varphi_n^k(X)^2 \frac{\sigma_n^k(x_1) - \rho_n(x_1)}{\sigma_n^k(x_1)} \chi_{E_n^k}(x_1) dX \\
  &= \frac{1}{N} \int_E (\sigma_n^k(x) - \rho_n(x)) dx,
\end{align*}
because the (first) marginal of $\varphi_n^k$ is $\sigma_n^k$, and $E \subseteq E_n^k$. Hence, using the inductive hypothesis,
\begin{align*}
 \int_{E} (\sigma_n^{k+1}(x) - \rho_n(x)) dx &\leq \left( 1 - \frac{1}{N} \right) \int_{E} (\sigma_n^{k}(x) - \rho_n(x)) dx \\
  & \leq \left( 1 - \frac{1}{N} \right)^{k+1} \int_E (\rho(x) - \rho_n(x)) dx.
\end{align*}
as wanted.
\end{proof}

The following proposition specifies that the sequence $(\varphi_n^k)_{k \geq 0}$ is not too far away from the target function $\varphi$ with respect to the $L^2$ topology.

\begin{prop} \label{k-estimate} For every $k \geq 0$,
\[
 \norm{\varphi}^2_{L^2(\rnd)} - \norm{\varphi_n^k}^2_{L^2(\rnd)} \leq 2N \norm{\sqrt{\rho} - \sqrt{\rho_n}}_{L^2(\R^d)}.
\]

\end{prop}
\begin{proof}
We denote for simplicity the $L^2$-norm as $\norm{\cdot}$ both on $\rnd$ and on $\R^d$, since there cannot be any confusion. By the definition of the $\varphi_n^j$'s  we may compute for every $j \geq 0$
\[
 \norm{\varphi_n^{j+1}}^2 = \int \varphi_n^{j}(X)^2 (1 - S_n^j(X)) dX = \norm{\varphi_n^j}^2 - \int_{E_n^j} (\sigma_n^j(x) - \rho_n(x)) dx.
\]
Hence, using Lemma \ref{int-E-lemma},
\begin{ieee*}{rCl}
 \norm{\varphi}^2 - \norm{\varphi_n^k}^2 & = & \sum_{j = 0}^{k-1} \left( \norm{\varphi_n^j}^2 - \norm{\varphi_n^{j+1}}^2 \right) = \sum _{j=0}^{k-1}  \int_{E_n^j} (\sigma_n^j - \rho_n) \\
  & \leq & \sum _{j=0}^{k-1} \left( \frac{N-1}{N} \right)^ j \int_{E_n^j} (\rho - \rho_n) = \sum _{j=0}^{k-1} \left( \frac{N-1}{N} \right)^ j \int_{E_n^j} |\rho - \rho_n|\\
  & \leq & \sum _{j=0}^{k-1} \left( \frac{N-1}{N} \right)^ j \int |\rho - \rho_n| \leq N \int |\rho - \rho_n |,
\end{ieee*}

Now the H\"{o}lder inequality and the elementary estimate $(\sqrt{a} + \sqrt{b})^2 \leq 2(a + b)$ lead to
\begin{align*}
 \int \abs{\rho - \rho_n} &\leq \left( \int \abs{\sqrt{\rho} + \sqrt{\rho_n}}^2 \right)^{1/2} \left( \int \abs{\sqrt{\rho} - \sqrt{\rho_n}}^2 \right)^{1/2} \\
  &=  \left( 2\int (\rho + \rho_n) \right)^{1/2} \norm{\sqrt{\rho} - \sqrt{\rho_n}} = 2 \norm{\sqrt{\rho} - \sqrt{\rho_n}}. \qedhere
\end{align*}
\end{proof}

We are ready to define the functions $\varphi_n$. Let

\begin{align*}
 \varphi_n^{\infty}(X) &= \lim_{k \to \infty} \varphi_n^k \\
 \sigma_n^{\infty}(x) &= \int \varphi_n^{\infty}(x, x_2, \dotsc, x_N)^2 dx_2 \dotsm dx_N.
\end{align*}

They are well defined due to Lemma \ref{monotonicity-lemma}, and $\varphi_n^\infty$ is symmetric, since it is the pointwise limit of symmetric functions; let moreover
\begin{align*}
 q_n &= \int \left( \rho_n(x) - \sigma_n^{\infty}(x) \right) dx \\
 \alpha_n(X) &= \frac{1}{q_n^{N-1}} \prod_{j = 1}^N (\rho_n(x_j) - \sigma_n^{\infty}(x_j)) \\
 \varphi_n(X) &= \sqrt{\varphi_n^{\infty}(X)^2 + \alpha_n(X)}.
\end{align*}
where the second term is set to zero if $q_n = 0$. Observe that the function $\varphi_n$ is symmetric, because $\alpha_n$ is symmetric by construction. The definition is well-posed since $\alpha_n$ is non-negative, as proved in the following

\begin{lemma} \label{sigma-lemma} $\rho_n(x) - \sigma_n^{\infty}(x) \geq 0$.
\end{lemma}

\begin{proof} Using that $\frac{\sigma_n^k - \rho_n}{\sigma_n^{k}} \chi_{E_n^{k}} \geq 0$, one has
\begin{align*}
 \rho_n(x) - \sigma_n^{k+1}(x) &= \rho_n(x) - \int_{(\R^d)^{N-1}} \varphi_n^{k+1}(X)^2 dx_2 \dotsm dx_N \\
  &= \rho_n(x) - \sigma_n^k(x) + \int_{\R^{(N-1)d}} \varphi_n^k(X)^2 S_n^k(X) dx_2 \dotsm dx_N \\
  &\geq (\rho_n(x) - \sigma_n^k(x)) \left( 1 - \frac{1}{N} \chi_{E_n^{k}}(x) \right)
\end{align*}

If $x \in E_n^k$ then
\[
 (\rho_n(x) - \sigma_n^k(x)) \left( 1 - \frac{1}{N} \chi_{E_n^k}(x) \right) = \frac{N-1}{N} (\rho_n(x) - \sigma_n^k(x));
\]
on the other hand, if $x \in (E_n^k)^c$, then $\rho_n(x) - \sigma_n^k(x) \geq 0$, and hence
\[
 (\rho_n(x) - \sigma_n^k(x)) \left( 1 - \frac{1}{N} \chi_{E_n^k}(x) \right) = \rho_n(x) - \sigma_n^k(x) \geq \frac{N-1}{N} (\rho_n(x) - \sigma_n^k(x)).
\]
So that for every $x \in \R^d$,
\[
 \rho_n(x) - \sigma_n^{k+1}(x) \geq \frac{N-1}{N} (\rho_n(x) - \sigma_n^k(x)),
\]
and letting $k \to \infty$,
\[
 \rho_n(x) - \sigma_n^{\infty}(x) \geq \frac{N-1}{N} (\rho_n(x) - \sigma_n^{\infty}(x)) \implies \rho_n(x) - \sigma_n^{\infty}(x) \geq 0. \qedhere
\]
\end{proof}

Finally, $\varphi_n \to \varphi$ in $L^2$ as $n$ goes to $\infty$, as is proved in the following

\begin{prop}
 \[
  \norm{\varphi_n - \varphi}^2_{L^2} \leq 2(2N+1) \norm{\sqrt{\rho} - \sqrt{\rho_n}}_{L^2}
 \]
\end{prop}

\begin{proof}
By the monotonicity described in Lemma \ref{monotonicity-lemma}, $\varphi \geq \varphi_n^\infty$ and then 
 \begin{align*}
  \abs{\varphi_n - \varphi}^2 &= \varphi^2 + \left( \varphi_n^\infty \right)^2 + \alpha_n - 2 \varphi \sqrt{\left( \varphi_n^\infty \right) ^2 + \alpha_n} \\
    &\leq \varphi^2 + \left( \varphi_n^\infty \right)^2 + \alpha_n - 2 \varphi \varphi_n^\infty \leq  \varphi^2 - \left( \varphi_n^\infty \right)^2 + \alpha_n
 \end{align*}

 Integrating over $\R^d$ leads to
 \[
  \norm{\varphi_n - \varphi}^2_{L^2} \leq \norm{\varphi}^2_{L^2} - \norm{\varphi_n^\infty}^2_{L^2} + \norm{\alpha_n}_{L^1}.
 \]
 
 Letting $k \to \infty$ in Proposition \ref{k-estimate} and using the monotone convergence theorem, one has
 \[
  \norm{\varphi}^2_{L^2} - \norm{\varphi_n^\infty}^2_{L^2} \leq 2N \norm{\sqrt{\rho} - \sqrt{\rho_n}}_{L^2}.
 \]
 
 On the other hand, recalling the final step of the proof of Proposition \ref{k-estimate} and using again the monotone convergence theorem,
 \begin{align*}
  \norm{\alpha_n}_{L^1} &= \int \left( \rho_n(x) - \sigma_n^\infty(x) \right) dx \\
    &\leq \int \abs{\rho_n(x) - \rho(x)} dx + \int \left( \rho(x) - \sigma_n^\infty(x) \right) dx \\
    &\leq 2 \norm{\sqrt{\rho} - \sqrt{\rho_n}}_{L^2} + \int \varphi(X)^2 dX - \int \varphi_n^\infty(X)^2 dX \\
    &= 2 \norm{\sqrt{\rho} - \sqrt{\rho_n}}_{L^2} + \norm{\varphi}^2_{L^2} - \norm{\varphi_n^\infty}^2_{L^2} \\
    &\leq 2(N+1) \norm{\sqrt{\rho} - \sqrt{\rho_n}}_{L^2}. \qedhere
 \end{align*}
\end{proof}

This concludes the proof of Theorem \ref{main-thm}.

\section{Sobolev regularity operators} \label{H1-regularity}

%
%
%

Let $\rho_n, \rho, \varphi_n$ and $\varphi$ as in Section \ref{L2-sequence} and assume, additionally, that $\varphi \in H^1$ and $\sqrt{\rho_n} \to \sqrt{\rho}$ in $H^1(\R^d)$. The sequence $(\varphi_n)$ constructed in Section \ref{L2-sequence} is such that $\varphi_n \in L^2(\rnd)$ with $\varphi_n \to \varphi$ in $L^2(\rnd)$.
We will now improve the regularity and the convergence of $(\varphi_n)$ at least up to $H^1(\rnd)$. In order to do so, we introduce the Gaussian mollifiers
\begin{align}
 \eta^\eps(z) &= \frac{1}{(2\pi\eps)^{d/2}} \exp \left( -\frac{\abs{z}^2}{2\eps} \right), \quad z \in \R^d \nonumber \\
 \phi^\eps(Z) &= \prod_{j = 1}^N \eta_\eps(z_j), \quad Z = (z_1, \dotsc, z_N) \in \rnd. \label{gaussian-kernel}
\end{align}

Given any function $u \in L^2(\rnd)$, first we regularize the measure $\abs{u}^2 dX$ by convolution
\[\Lambda^\eps[u](Y) = \int \abs{u(X)}^2 \phi^\eps(Y-X) dX,\]
and consider the marginal of this regularization
\[ \rho^\eps[u](y) = \int \Lambda^\eps[u](y,y_2,\dotsc,y_N) dy_2 \dotsm dy_N. \]
Then using the kernels 
\[ P^\eps[u](X,Y) = \prod_{j = 1}^N \frac{\eta^\eps(y_j-x_j) \rho[u](x_j)}{\rho^\eps[u](y_j)}, \]
\[ \Gamma^\eps[u](X,Y) = \Lambda^\eps[u](Y) P^\eps[u](X,Y) \]
we define
\[ \Theta^\eps[u](X) = \int \Gamma^\eps[u](X,Y) dY.\]

Let $u^\eps:= ( \Theta^\eps[u])^{1/2}$. This construction, in a more general setting, has been presented and studied in detail in \cite{bindini2019smoothing}, where it is also proved that $u^\eps \in H^1(\rnd)$ for every $\eps > 0$.

\begin{rema}
	The operator  $\sqrt{\ } \circ \Theta^\eps$ is builded out of four passages. We start from $u$, we consider the measure $|u|^2 dX$, we regularize this last measure and, since the regularization process perturbs the marginals, we bring the marginals back to the originals. Finally we take the square root. 
	
	The construction to bring back the marginals makes use of the idea of composition of transport plans introduced in \cite{ambrosio2008gradient}
\end{rema}

If $u$ is a symmetric function, \ie, $u(x_1, \dotsc, x_N) = u(x_{\sigma(1)}, \dotsc, x_{\sigma(N)})$ for every permutation $\sigma \in \mathfrak{S}_N$, the construction defined above gives a symmetric function $u^\eps$. Moreover, it ``reallocates'' the marginals as described in the following Lemma:

\begin{lemma}[Marginal properties] \label{basic-properties} For every $\eps > 0$ the following hold.
\begin{enumerate}[(i)]
 \item $\rho^\eps[u] = \rho[u] * \eta^\eps;$
 \item $\rho[u^\eps] = \rho[u]$.
\end{enumerate}
\end{lemma}

\begin{proof} 
  \begin{enumerate}[(i)]
    \item By the Fubini's theorem,
    \begin{align*}
     \rho^\eps[u](y) &= \int \Lambda^\eps[u](y,y_2,\dotsc,y_N) dy_2 \dotsm dy_N \\
      &= \int \abs{u(X)}^2 \eta^\eps(y-x_1) \prod_{j = 2}^N \eta^\eps(y_j-x_j) dX dy_2 \dotsm dy_N \\
      &= \int \rho[u](x_1) \eta^\eps(y-x_1) \prod_{j = 2}^N \eta^\eps(y_j-x_j) dx_1 dy_2 \dotsm dy_N \\
      &= \int \rho[u](x_1) \eta^\eps(y-x_1) dx_1 = (\rho[u] * \eta^\eps)(y).
    \end{align*}
    
    \item Using (i), from the definition of $\Theta^\eps$ and the Fubini's theorem we get
      \begin{align*}
     \int_{\R^{(N-1)d}} \abs{u^\eps}^2 dx_2 \dotsm dx_N &= \int_{\R^{(N-1)d}} \Theta^\eps[u](X) dx_2 \dotsm dx_N \\
      &= \int_{\R^{(N-1)d} \times \rnd} \Gamma^\eps[u](X,Y) dY dx_2 \dotsm dx_N \\
      &= \int_{\rnd} \Lambda^\eps[u](Y) \frac{\eta^\eps(y_1-x_1)\rho[u](x_1)}{\rho^\eps[u](y_1)} dY \\
      &= \int_{\R^d} \rho^\eps[u](y_1) \frac{\eta^\eps(y_1-x_1)\rho[u](x_1)}{\rho^\eps[u](y_1)} dy_1 \\
      &= \rho[u](x_1).
    \end{align*}	
  \end{enumerate}
 
\end{proof}

Moreover, the following convergence properties are enjoied by the operator $\Theta^\eps$ and are studied and proved, in a more general setting, in \cite{bindini2019smoothing}

\begin{thm}[{\cite[Theorem 5.1]{bindini2019smoothing}}] \label{H1-H1-continuity} Let $u_n \to u$ in $L^2(\rnd)$ and $\sqrt{\rho[u_n]} \to \sqrt{\rho[u]}$ in $H^1(\R^d)$. Then, for every $\eps > 0$, $u_n^\eps \to u^\eps$ in $H^1(\rnd)$.
\end{thm}

\begin{thm}[{\cite[Theorem 7.6]{bindini2019smoothing}}] \label{eps-continuity} Let $u \in H^1(\rnd)$ non-negative. Then $u^\eps \to u$ in $H^1(\rnd)$ as $\eps \to 0$.
\end{thm}

We are now able to prove the following

\begin{thm} \label{sym-H1} Let $\rho_n,\rho \in \mathcal{R}$ such that $\sqrt{\rho_n} \to \sqrt{\rho}$ in $H^1(\R^d)$, and let $\varphi \in H^1(\rnd)$ symmetric and non-negative be such that $\rho[\varphi] = \rho$. Then there exist $u_n \in H^1(\rnd)$ such that $u_n \to \varphi$ in $H^1(\rnd)$ and $\rho[u_n] = \rho_n$.
\end{thm}

\begin{proof}
  Let $\varphi_n$, $\varphi$ defined in Section \ref{L2-sequence}: the idea is to take a suitable diagonal sequence $u_n \eqdef (\Theta^{\eps(n)}[\varphi_n])^{1/2}$. Let $N_0 = 1$, and for $k \geq 1$ choose $N_k \in \N$ such that
  \begin{enumerate}[(i)]
    \item $N_k > N_{k-1};$
    \item $\norm{\sqrt{\Theta^{2^{-k}}[\varphi_n]} - \sqrt{\Theta^{2^{-k}}[\varphi]}} \leq 2^{-k}$ for every $n \geq N_k$.
 \end{enumerate}
  
 The sequence $(N_k)_{k \geq 0}$ is well defined due to Theorem \ref{H1-H1-continuity} and increasing. Given $n \geq 1$, let $k$ be such that $N_k \leq n < N_{k+1}$, and set $\eps(n) = 2^{-k}$. When $N_k \leq n < N_{k+1}$, by construction we have
 \begin{ieee*}{rCl}
   \norm{\sqrt{\Theta^{\eps(n)}[\varphi_n]} - \varphi} & \leq & \norm{\sqrt{\Theta^{2^{-k}}[\varphi_n]} - \sqrt{\Theta^{2^{-k}}[\varphi]}} + \norm{\sqrt{\Theta^{2^{-k}}[\varphi]} - \varphi} \\
   & \leq & 2^{-k} + \norm{\sqrt{\Theta^{2^{-k}}[\varphi]} - \varphi}.
 \end{ieee*}
 
 As $n \to \infty$, also $k \to \infty$ and the right-hand side goes to zero due to Theorem \ref{eps-continuity}.
\end{proof}

To avoid any confusion, in the following we shall denote by $(\varphi_n)$ a sequence such that $\varphi_n \to \varphi$ in $H^1(\rnd)$ and $\rho[\varphi_n] = \rho_n$.

\begin{rema}
  If the original wave-function was symmetric and non-negative, then by Theorem \ref{sym-H1} we already get the desired approximating wave-functions which are also symmetric and non-negative, thus proving Theorem \ref{main-thm-1}.
\end{rema}


\section{Approximation with signs} \label{sign-section}

 Let $\rho_n, \rho, \varphi_n$ and $\varphi$ like in the previous sections, and assume now that $\varphi = \abs{\psi}$, where $\psi \in H^1((\R^d)^N; \R)$. In this section, starting from $\varphi_n$, we will construct $\psi_n \in H^1 ((\R^d)^N; \C)$ such that
 $\psi_n \mapsto \rho_n$ and $\psi_n \stackrel{H^1}{\to} \psi$. Some weighted Sobolev spaces will be the main tool of the construction. 
 
To every measurable $\lambda: \rnd \to \R^n$ (scalar or vectorial) we may associate some spaces related to the measure $|\lambda|^2 (X) dX$. In particular this will be used for $\lambda$ equal to the wave function $\psi$ or equal to the gradient of the wave function $\nabla \psi$.  
The most natural space is 
$$ L^2 (|\lambda|^2 dX; \C):= \gra{f \colon \rnd \to \C \st \int |f(X)|^2 |\lambda (X)|^2 dX < +\infty}.$$
When $\lambda  \in H^1 (\rnd; \R)$ we may define also the Sobolev spaces relative to the measure $|\lambda| ^2 dX = \lambda^2 dX$.   
First  we need a definition of the gradient: 
\begin{defi}  If $f \in L^2(\lambda^2 dX)$, the gradient $\onab f$ is defined by the identity
\begin{equation}
\int \onab f \varphi \lambda^2 dX = - \int f \nabla \varphi \lambda ^2 dX - 2 \int f  \varphi \frac{\nabla \lambda}{\lambda} \lambda^2 dX \quad \forall \varphi \in \Con^\infty _c.
\end{equation}
\end{defi}
It is then natural to define the Sobolev space
$$ H^1 (\lambda^2 dX):= \gra{f \in L^2 (\lambda^2 dX) \st  \int |\onab f(X)|^2 \lambda^2 dX < +\infty}.$$
\begin{rema} \label{freg}
  If $f$ is a $\Con^1$ function then $\onab f = \nabla f$, the usual gradient. Indeed, if $\varphi \in \mathcal{C}_c^\infty$, then
  \begin{ieee*}{rCl}
    \int \onab f \varphi \lambda^2 dX & = & - \int f \nabla \varphi \lambda ^2 dX - 2 \int f  \varphi \frac{\nabla \lambda}{\lambda} \lambda^2 dX \\
    & = & \int \nabla f \varphi \lambda ^2 dX + 2 \int f  \varphi \frac{\nabla \lambda}{\lambda} \lambda^2 dX - 2 \int f  \varphi \frac{\nabla \lambda}{\lambda} \lambda^2 dX \\
    & = & \int \nabla f \varphi \lambda ^2 dX.
  \end{ieee*} 
\end{rema}

\begin{rema} \label{Leibniz}
  If $f,g \in H^1(\lambda^2dX)$, and $fg, f\onab g, g\onab f \in L^2(\lambda^2dX)$, then $fg \in H^1(\lambda^2dX)$ and $\onab (fg) = f\onab g + g\onab f$. This is Corollary 2.6 in \cite{cattiaux1996entropy}.
\end{rema}


The next construction relates these ideas to the objects we know from the previous sections. Let $\psi \in H^1 (\rnd) $ so that also $|\psi| \in H^1 ( \rnd)$; then there exists a measurable function $e: \rnd \to \{-1, 1\}$ such that 
$\psi = e |\psi|$.  The function $e$ coincides almost everywhere with $ \psi / |\psi| $ in the set where $\psi \neq 0$. From now on, let $\lambda = \abs{\psi}$.

\begin{lemma}\label{regolarite}  It holds that $e \in H^1 (|\psi|^2 dX)$ and $ \onab e= 0$ $|\psi|^2 dX-a.e.$ Moreover, since $|e|\leq 1$,  $e \in L^2 (|\nabla \psi|^2 dX)$. 
\end{lemma} 
\begin{proof} Since $|e|=1$ $|\psi|^2 dX-a.e.$,  $e \in L^2 (|\psi|^2 dX).$  Let $\varphi \in \Con^\infty _c$,
$$ \int \onab e \varphi |\psi|^2 d X = - \int e \nabla \varphi |\psi|^2 dX - 2 \int e \varphi |\psi| \frac{\psi}{|\psi|} \nabla \psi dX=0, $$  
since 
$$\int e \nabla \varphi |\psi|^2 dX = \int \nabla \varphi \psi |\psi|dX= -  \int \varphi \nabla \psi |\psi| dX  - \int \varphi  \psi \frac{\psi}{|\psi|} \nabla \psi dX = -2  \int \varphi  |\psi| \nabla \psi dX .$$
\end{proof}

We are interested in smooth approximations in these Sobolev spaces. This question is well studied in the literature. The following is a consequence of \cite[Theorem 2.7]{cattiaux1996entropy}; we invite the reader to see also the references in that paper for a more complete picture.
 
\begin{thm}\label{smoothe} There exists a sequence $\{e_n\} \in \Con^\infty \cap H^1(|\psi|^2 dX)$ such that 
\begin{enumerate}[i)]
\item $|e_n| \leq 1$, 
\item $e_n \to e$ in $ H^1 (|\psi|^2 dX)$, 
\item $ e_n \to e$ in $L^2 (| \nabla \psi|^2 dX)$.
\end{enumerate}
\end{thm}
\begin{proof}
	Choose a sequence of smooth cut-off functions $(c_n)_{n \geq 1}$ such that:
  \begin{itemize}
    \item $0 \leq c_n \leq 1$;
    \item $c_n \equiv 1$ on $B(0,n-1)$, $c_n \equiv 0$ on $B(0,n)^c$;
    \item $\Lip(c_n) \leq 2$.
  \end{itemize}
  
  First we consider $e\cdot c_n$ and we prove that it satisfies the properties i)-iii) above. The property i) is obvious. Also the $L^2$-convergence is easy:
  \[ \int \abs{ec_n - e} \lambda^2 dX \leq \int_{B(0,n)^c} \lambda^2 dX \quad \text{and} \quad \int \abs{ec_n - e}\abs{\nabla\lambda}^2 dX \leq \int_{B(0,n)^c} \abs{\nabla\lambda}^2 dX  \]
  converge to zero since $\lambda, \nabla \lambda \in L^2(\R^{Nd})$. Combining Remark \ref{freg} and Remark \ref{Leibniz} we have $\onab (ec_n) = c_n \onab e + e \onab c_n = e \onab c_n = e \nabla c_n$, and we must prove that it converge to 0 in $L^2(\lambda^2 dX)$. Indeed we have
  \[ \int \abs{e \nabla c_n}^2 \lambda^2 dX \leq 4 \int_{B(0,n-1)^c} \lambda^2 dX. \]
  
  Now the second step is to regularize by convolution with a standard mollifier of compact support $J_\eps$ defined by $J_\eps(X) = 1/\eps^{Nd} J(X/\eps)$, where $J$ is non-negative and supported in the unit ball, with $\int J = 1$. It is shown in \cite[Theorem 2.7]{cattiaux1996entropy} that $J_\eps * (ec_n)$ converge to $ec_n$ for fixed $n$ as $\eps \to 0$. Thus it suffices to take $\eps_n$ small enough so that $\norm{J_{\eps_n} * (ec_n) - ec_n}_{H^1(\lambda^2 dX)}$ converges to zero to conclude.
\end{proof}

\begin{rema} \label{symm-e}
  If the function $e$ is symmetric, it is possible to make $e_n$ to be symmetric as well. It suffices to choose $c_n$ (the cut-off functions) to be symmetric. Then the process of convolution maintains symmetry if the kernel is symmetric (as for instance the one defined by \eqref{gaussian-kernel}).
\end{rema}

In order to have a good behaviour of the approximating sequence $(e_n)_{n\geq 1}$ for the estimates that will be needed in the proof of Theorem \ref{final-thm}, we must also control the Lipschitz constant of $e_n$. This may be done as a consequence of the following

\begin{lemma} \label{subsequence}
  Given sequences of non-negative real numbers $(M_n)$ and $(a_k)$ such that $a_k \to 0$, there exists a choice $(n_k)$ of indexes such that
  \begin{enumerate}[(i)]
    \item $n_k \nearrow +\infty$;
    \item $M_{n_k} a_k \to 0$.
  \end{enumerate}
\end{lemma}

\begin{proof}
  Given $n$, let $K(n)$ such that $M_na_k < 2^{-n}$ for all $k \geq K(n)$, and choose also $K(n+1) > K(n)$. Now we define the sequence $(n_k)$ as follows:
  \[ n_k = \begin{cases}
  1 & \text{if } k < K(1) \\ n & \text{if } K(n) \leq k < K(n+1).
  \end{cases} \]
  
  By construction we have $M_{n_k}a_k < 2^{-n}$ for all $k \geq K(n)$, thus proving (ii). On the other hand, given $L \in \N$, if $k \geq K(L)$ we have $n_k \geq L$, which proves (i).
\end{proof}

\begin{coro} \label{LipCoro}
  Given $(a_n)$ such that $a_n \to 0$, the sequence in Theorem \ref{smoothe} may be chosen such that $\Lip(e_n)a_n \to 0$.
\end{coro}

\begin{proof}
  Apply Lemma \ref{subsequence} with $M_n = \Lip (e_n)$ to select a suitable sequence $(e_{n_k})$ with the desired property.
\end{proof}

\begin{defi}\label{lifting} Let $\omega \in \Con ^1 ([-1,1], S^1_+)$  be defined by
	$$ s \mapsto e^{i(1-s)\frac{\pi}{2}}.$$
\end{defi}
The function $\omega$ is such that $\abs{\omega} = 1$, $\omega(-1) = -1$ and $\omega(1) = 1$ so that $\omega(e(x)) =e(x)$ a.e. in the set $\psi \neq 0$. Moreover, observe that $|\omega'|= \frac{\pi}{2}$ and $|\omega(s)-\omega(t)| \leq \frac{\pi}{2} \abs{s-t}$ for all $s,t \in [-1,1]$.
 
\begin{thm} \label{final-thm} Let $\omega$ be the function defined above. Let $e_n \in \Con^\infty_c $ with values in  $[-1,1]$ be such that $e_n \to e$ in $ H^1 (|\psi|^2 dX)$ and 
$ L^2 (|\nabla \psi|^2 dX)$.  Then $\psi_n:= \omega(e_n) \varphi_n \to \psi$ in $H^1$.
\end{thm}
\begin{proof}  First the $L^2$ convergence which is easier. 
\begin{eqnarray*}
\|\psi_n-\psi\|_{L^2} &=& \|\omega(e_n) \varphi_n- e|\psi|\| _{L^2}\leq \|\omega(e_n) \varphi_n - \omega(e_n) |\psi| \|_{L^2} + \| \omega(e_n) |\psi| - e |\psi| \|_{L^2}\\
&= & \| \varphi_n -|\psi| \| _{L^2} + \| \omega(e_n) - e \|_{L^2 (|\psi|^2 dX)}\\
&\leq &  \| \varphi_n -|\psi| \| _{L^2}+ \| \omega(e_n) - \omega(e) \|_{L^2 (|\psi|^2 dX)} +
\| \omega(e) - e \|_{L^2 (|\psi|^2 dX)} \\ 
&\leq &  \| \varphi_n -|\psi| \| _{L^2}+ \frac{\pi}{2}\| e_n - e \|_{L^2 (|\psi|^2 dX)}.
\end{eqnarray*}
The last term converges to $0$ by Theorem \ref{smoothe} above.

For the $L^2$ convergence of gradients, let us first compute
$$ \nabla \psi_n= \omega'(e_n) \nabla e_n \varphi_n  + \omega (e_n) \nabla \varphi_n ,$$
$$ \nabla \psi= \nabla (e |\psi|) = e  \nabla |\psi|,$$
and in the second computation we used that $\nabla e=0$ a.e. where $|\psi| \neq 0$.
\begin{eqnarray*}
\|\nabla \psi_n- \nabla \psi \|_{L^2} &=& \|\omega' (e_n) \nabla e_n \varphi_n + \omega(e_n) \nabla \varphi_n - e \nabla | \psi \|\\
& \leq &  \|\omega'(e_n) \nabla e_n \varphi_n - \omega'(e_n) \nabla e_n |\psi| \| + \| \omega'(e_n) \nabla e_n |\psi|\| \\
& & {} + \| \omega(e_n) \nabla \varphi_n - e \nabla |\psi|\|. 
\end{eqnarray*}
The three terms on the right-hand-side above may be studied separately, the most difficult one being the first. We have 
$$ \|\omega'(e_n) \nabla e_n \varphi_n - \omega'(e_n) \nabla e_n |\psi| \|^2 \leq \frac{\pi}{2} \int |\nabla e_n|^2 |\varphi_n -|\psi||^2 dX \leq \frac{\pi}{2} \Lip (e_n)^2 \int |\varphi_n -|\psi||^2 dX.$$
The last term of the inequality converges to $0$ if we choose $a_n= \norm{\varphi_n - |\psi|}_{L^2}$ in Corollary \ref{LipCoro}.
 
The second term 
\[\| \omega'(e_n) \nabla e_n |\psi|\|^2 \leq \frac{\pi}{2} \int |\nabla e_n|^2 |\psi|^2 dX \]
and this goes to $0$ by Theorem \ref{smoothe} ii) and Lemma \ref{regolarite}.
Finally we control the third term by breaking it down again.
\begin{eqnarray*}
	\| \omega(e_n) \nabla \varphi_n - e \nabla |\psi|\|_{L^2} &\leq& \| \omega(e_n) \nabla \varphi_n - \omega(e_n) \nabla |\psi|\|_{L^2} + \| \omega(e_n) \nabla |\psi| -e \nabla |\psi| \|_{L^2} \\
	&\leq & \| \nabla \varphi_n - \nabla |\psi|\|_{L^2} + \| \omega(e_n) \nabla |\psi|- \omega(e) \nabla |\psi| \|_{L^2}\\
	&\leq&   | \nabla \varphi_n - \nabla |\psi|\|_{L^2} +  \frac{\pi}{2}\| e_n-e\|_{L^2  (|\nabla \psi|^2 dX)}.
\end{eqnarray*}
The last term converges to $0$ by the convergence of $\varphi_n$ to $|\psi|$ and by Theorem \ref{smoothe} iii). 
\end{proof}

In conclusion, notice that the approximating sequence built in this way maintains the symmetry property. Indeed, $\varphi_n$ is symmetric for every $n$, and so is the sign function $e$. By Remark \ref{symm-e} we may choose $e_n$ to be symmetric. Finally, if $e_n$ is symmetric, so is $\omega(e_n)$, and hence $\omega(e_n)\varphi_n$ is symmetric, finally proving Theorem \ref{main-thm-2}.




\bibliographystyle{plain}

\nocite{ambrosio2008gradient, bindini2017optimal, bindini2019smoothing, buttazzo2018continuity, buttazzo2012optimal, cattiaux1996entropy, cotar2013density, cotar2018smoothing, hohenberg1964inhomogeneous, kohn1965self, lewin2018semi, levy1979universal, lieb2002density}

\bibliography{biblio.bib}

\end{document}